\documentclass[9pt,twocolumn,twoside,dvipsnames]{pnas-new}

\usepackage{amsfonts}
\usepackage{amsthm}
\usepackage{amsmath}
\usepackage{parskip}
\usepackage{mathrsfs}
\usepackage{amssymb}
\usepackage{xtab}
\usepackage{setspace}
\usepackage{datetime}
\usepackage{svg}
\usepackage{csquotes}

\usepackage[dvipsnames]{xcolor}

\usepackage[utf8]{inputenc}
\usepackage{amssymb,mathalfa,float,array,titlesec,mdframed,listings,amsmath,booktabs,dsfont,textcomp,graphicx,wrapfig,cases}
\usepackage[font=small,labelfont=bf]{caption}

\usepackage[utf8]{inputenc}
\usepackage[english]{babel}

\theoremstyle{remark}
\newtheorem*{remark}{Remark}
\theoremstyle{definition}

\newtheorem{proposition}{Proposition}

\usepackage{float}

\usepackage{placeins}
\usepackage{subcaption}
\usepackage{hyperref}
\usepackage{multirow}

\graphicspath{{Figures/}}

\templatetype{pnasresearcharticle}

\title{Cyclists' Cardiac Conundrum}

\author[a]{Andrew Nugent}
\author[a]{Yi Ting Loo}
\author[a]{Jack Buckingham}
\affil[a]{Mathematics of Real World Systems II CDT, University of Warwick}

\significancestatement{In recent years there has been increasing evidence that those who frequently participate in high intensity endurance training have an elevated risk of developing Atrial Fibrillation (AF), a type of irregularity in heart rhythm. Although an electrocardiogram (ECG) is available for accurate diagnosis of AF, common sports devices worn by athletes only collect heart rate data in beats per minute (bpm) during an activity, usually sampled at every second which generates a time series. This project analyses characteristics of these time series and their correlation with diagnoses of AF and aims to provide useful indicators of warning signs of irregularity in heart rhythm to an athlete which can prompt earlier diagnosis, and subsequently treatment, of AF.
}

\authordeclaration{No conflicts of interest are declared.}

\keywords{Arrhythmia $|$ Heart rate data $|$ Spike detection $|$ Poisson process} 

\begin{abstract}
Arrhythmia is an abnormality of the heart's rhythm, caused by problems in the conductive system and resulting in irregular heartbeats. There is increasing evidence that undertaking frequent endurance sports training elevates one’s risk of arrhythmia. Arrhythmia is diagnosed using an electrocardiogram (ECG) but this is not typically available to athletes while exercising. Previous research by Crickles investigates the usefulness of commonly available heart rate data in detecting signs of arrhythmia. It is hypothesised that a feature termed `gappiness', defined by jumps in the heart rate while the athlete is under exertion, may be a characteristic of arrhythmia. A correlation was found between the proportion of `gappy' activities and survey responses about heart rhythm problems. We develop on this measure by exploring various methods to detect spikes in heart rate data, allowing us to describe the extent of irregularity in an activity via the rate of spikes. We first compare the performance of these methods on simulated data, where we find that smoothing using a moving average and setting a constant threshold on the residuals is most effective. This method was then implemented on real data provided by Crickles from 168 athletes, where no significant correlation was found between the spike rates and survey responses. However, when considering only those spikes that occur above a heart rate of 160 beats per minute (bpm) a significant correlation was found. This supports the hypothesis that jumps at only high heart rates are informative of arrhythmia and indicates the need for further research into better measures to characterise features of heart rate data. 

\end{abstract}

\renewcommand{\today}{Thursday 9\textsuperscript{th}June, 2022}
\dates{This manuscript is made publicly available under a CC BY license: \url{https://creativecommons.org/licenses/by/4.0/}}

\begin{document}

\maketitle
\thispagestyle{firststyle}
\ifthenelse{\boolean{shortarticle}}{\ifthenelse{\boolean{singlecolumn}}{\abscontentformatted}{\abscontent}}{}

\section{Introduction} \label{sec: introduction}

Arrhythmia is an abnormality of the heart's rhythm, the most common type of which is atrial fibrillation (AF) \cite{NHSarrhythmia}. A diagnosis is confirmed by observing an episode of AF during an electrocardiogram (ECG) \cite{developed2010guidelines}, which can be aided by patients wearing a portable ECG recorder \cite{NHSarrhythmia,camm2012usefulness}. While moderate physical exercise is associated with a reduced risk of AF there is some evidence that strenuous endurance exercise increases an individual's risk \cite{aizer2009relation,centurion2019association}. As such there has been an interest in attempts to use ECGs in smart watches, worn by athletes to record information about their exercise, to detect AF \cite{mcmanus2013novel,li2019current}. As these devices with ECGs can be expensive, it could be beneficial to instead determine if commonly available heart rate data contains any informative signs of AF. 

The Crickles project \cite{cricklesWebsite} was established to provide athletes with cardiac stress analysis, calculated using data from their Strava activities, and inform them about their training load relative to other athletes. This data has been used by Ian Green and Mark Dayer to investigate the possibility of using heart rate data to detect the presence of heart rhythm problems \cite{cricklesPaper}. They hypothesised that jumps at high heart rates, for example jumping from 180 to 195 beats per minute (bpm) then immediately dropping again, or jumps to unrealistically high values, could be informative. As an athlete's heart rate monitor continuously samples current in the heart muscle it may plausibly be affected by arrhythmia. Green \textit{et al.} do not claim that such jumps are a direct observation of an episode of AF, nor that they are an accurate reflection of an athlete's heart rate, but rather that these observed jumps are errors in the heart rate monitor. Their key hypothesis is that the presence and/or frequency of such errors may still be informative.

In \cite{cricklesPaper} each activity is classified as `Regular', `Unclear', `Check\_Strap' or `Irregular', based on two key features. The feature of most interest is `gappiness' shown in Fig.\ref{fig:gappy heart rate}: an activity is gappy if it contains gaps in the range of observed heart rate readings over a threshold.
\begin{figure}[ht]
    \centering
    \includegraphics[width=\linewidth]{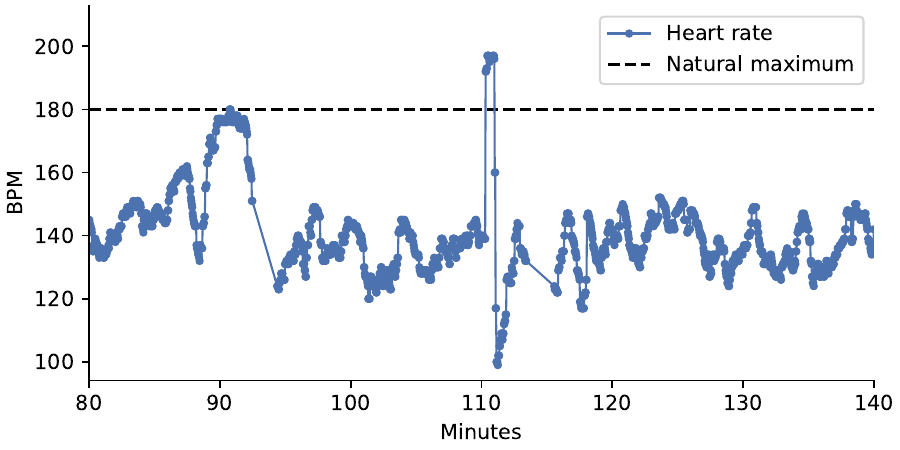}
    \caption{Time series showing a gappy heart rate, as the heart rate exceeds the natural maximum (dashed line). It is hypothesised that this gap are blips in the heart rate monitor caused by unusual electrical cardiac signals, and so a heightened frequency of these gaps for an athlete may be associated with arrhythmia.}
    \label{fig:gappy heart rate}
\end{figure}
This threshold is set as the maximum of the modal heart rate and lactate threshold heart rate (LTHR) of the athlete, which serves as an indication of when the athlete's heart rate is considered to be at a high range. Gappiness can be detected by marking heart rates as either `visited' or `unvisited' over the course of an activity. If there exists a visited heart rate higher than an unvisited heart rate and the threshold then a gap is observed. As the procedure for detecting gappiness uses the distribution of heart rates across an entire time series it provides a simple Boolean indicator of whether the whole activity is gappy or not. Activities can also be gappy at the bottom end of heart rate readings, using a mirror of the procedure above, but it is assumed that these gaps are not indicative of arrhythmia as the athlete is not under exertion. By design gappiness focuses on jumps that occur at the top end of the heart rate, potentially neglecting any informative behaviour that occurs when the heart rate is normal or relatively low.

\begin{figure}[ht]
    \centering
    \includegraphics[width=\linewidth]{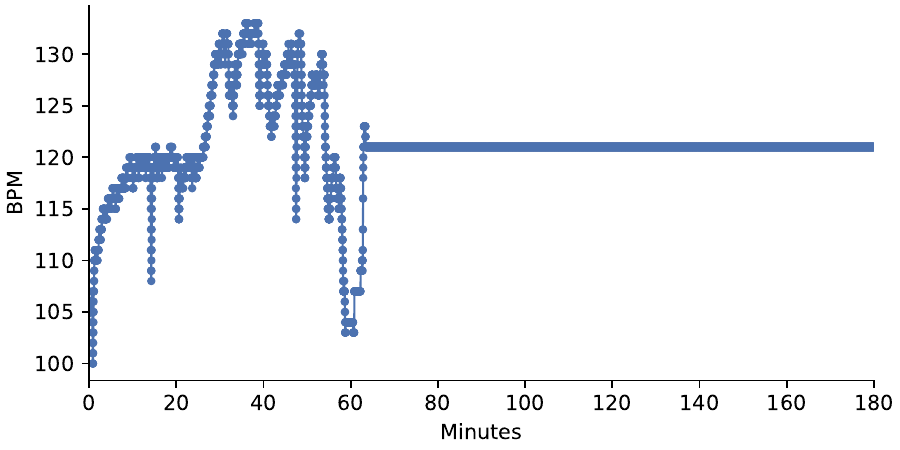}
    \caption{Time series showing a sticky heart rate, where the heart rate is held constant at 120 bpm for a long period of time which is physiologically implausible especially during endurance training, suggesting an error in the heart rate monitor.}
    \label{fig:sticky heart rate}
\end{figure}

For the second feature, an activity is called `sticky' if the heart rate is constant for at least 90 seconds, as seen in Fig.\ref{fig:sticky heart rate}. This indicates that the monitor is not correctly sampling the heart rate. Activities are then classified according to Table \ref{tab:activity classification}.

\begin{table}[H]
    \centering
    \caption{Classification of activities}
    \begin{tabular}{|c|c|c|c|}
        \cline{3-4}
        \multicolumn{2}{c|}{} & Non-sticky & Sticky \\ \hline
        \multirow{3}{3em}{Gappy} & High end only & \cellcolor[HTML]{ff9999} Irregular & \cellcolor[HTML]{e1e1ea} \\ \cline{2-2}
        & Low end only & \cellcolor[HTML]{e1e1ea} Unclear & \cellcolor[HTML]{e1e1ea} \\  \cline{2-2}
        & At both ends & \cellcolor[HTML]{e1e1ea} Unclear & \cellcolor[HTML]{e1e1ea} \\  \cline{1-2}
        \multicolumn{2}{|c|}{Non-gappy} & \cellcolor[HTML]{adebad} Regular & \multirow{-4}{*}{\cellcolor[HTML]{e1e1ea} Check\_Strap}\\ \hline
    \end{tabular}
    \label{tab:activity classification}
\end{table}
Activities classified as `check\_strap' or `unclear' were removed from consideration. This did not pose a problem as athletes typically recorded a large number of activities. For each athlete the \textit{irregular ratio} was defined as the percentage of remaining activities classified as irregular. 

The study in \cite{cricklesPaper} aimed to investigate the relationship between an athlete's irregular ratio and their cardiac health. Crickles users were invited to take part in a questionnaire, the key question of which was 
\begin{displayquote}
Do you have, or have you had a heart rhythm problem?
\end{displayquote}
The Yes/No response to this question was used to classify athletes into those with and without a diagnosed arrhythmia. The Crickles dataset was restricted to those athletes with a survey response. The dataset was restricted further by considering only cyclists, who were assumed to use a chest strap to monitor their heart rates. A statistically significant correlation of 0.24 was then found between the irregular ratio and responses to a heart health survey. 

Work so far on this problem describes the jumps in time series by labelling the entire activity `gappy' or `non-gappy'. We aim to expand on this by providing a more detailed characterisation of the level of irregularity for each activity/athlete, with the goal of strengthening the correlation previously observed. We begin by developing `spike detection' algorithms to locate jumps throughout the heart rate data, regardless of the range of the heart rate readings, and extract potentially important features such as the height of spikes and the heart rate at which they occur. We then infer the rate of spikes for each athlete and test for a correlation against reported heart rhythm problems. 

\section{Methods} \label{sec: methods}

We begin this section by describing a variety of methods to detect spikes in heart rate data, including smoothing, wavelet transforms and thresholding. 

We then introduce a method for simulating heart rate data. Creating simulated data was essential as privacy controls regarding data available to Crickles severely limited our own access. For exploratory analysis we had access to a small number of sample time series, but to utilise a larger dataset we needed the Crickles team to run and return any results. This restriction necessarily shaped our methodology: we first developed various spike detection methods, tuned their parameters, then compared their performance on simulated data to select a single method to apply to the real dataset. We outline our methods and comparison below, then close this section by describing the inference we perform on the detected spikes.

\subsection{Spike detection methods} \label{subsection:spike_detection_methods}

One method that could be used to detect spikes on time series data is to first detrend the time series \cite{techniques_for_detrending}, for example, by subtracting a smoothed version from the original. Then, a threshold can be set on the residuals, with times when the residuals exceed that threshold registered as spikes. The corresponding heart rate where the spikes occur can be found from the smoothed heart rate time series. The height of the spike can also be calculated by subtracting the smoothed heart rate from the original heart rate data where spikes occurred.

\subsubsection{Smoothing methods} 

In order to smooth the heart rate time series, one method was to utilize a simple moving average. Letting $\{X(t):t=1,\dots,T\}$ be the heart rate time series, define the smoothed time series, $\overline{X}(t)$ by
\begin{equation}
    \overline{X}(t) := \frac{1}{w} \sum_{k=-i}^{i} X(t+k)
\end{equation}
where $w=2i+1$ is called the width.

A more complex smoothing method finds a sparse representation in a discrete wavelet basis - a widely used technique for signal denoising \cite{intro_to_wavelets, mallat2009wavelets}. The mathematical theory of this method is explained in section \ref{SI: spike detection methods}.\ref{SI: spike detection - dwt} of the supplementary material.

\subsubsection{Thresholding methods} 
A simple thresholding method is to set a constant threshold by evaluating the $95$th percentile of the residuals, multiplied by a constant. 

Another method is to have an adaptive threshold \cite{adaptive_threshold_stackoverflow} on the residuals based on a rolling window, instead of evaluating a constant threshold from the entire set of residuals. For each time $t$ in the residual series, a window of width $w$ centred on time $t$ is considered and the threshold, $\text{Thresh}(t)$ is calculated by:

\begin{equation}
    \text{Thresh}(t) = \text{offset} + \mu(t) + \text{constant}*\sigma(t)
\end{equation}

where $\mu(t)$ and $\sigma(t)$ is the mean and standard deviation of residuals within the window respectively. The offset was introduced and set to a low value of $0.01$ to prevent any detection of spikes when there are regions of missing data in the heart rate time series. Indeed, the standard deviation $\sigma(t)$ should be zero in these regions, but fluctuations due to finite numerical precision meant that the residuals here are not exactly zero.

\subsubsection{Continuous wavelet transform}
The methods discussed so far detect spikes by thresholding the residuals after detrending the time series. An alternative approach instead uses the continuous wavelet transform (CWT). Loosely speaking, the CWT measures the similarity of the signal to the shape of a given `mother wavelet' at different timescales and locations. By thresholding these coefficients for a suitable, small timescale, we can identify locations in the time series which resemble spikes. Again, this threshold can be either constant or adaptive. The CWT is closely related to the discrete wavelet transform mentioned earlier and both transforms, along with the spike detection methods derived from them, are described in more detail in section \ref{SI: spike detection methods} of the supplementary material.

\subsection{Simulation algorithm} \label{sec: simulation algorithm}

Our simulation model is inspired by existing models for electricity prices, which are known to spike in a similar fashion \cite{mayer2015modeling,de2006nature}. These models separate the process into multiple components, each designed to capture a different feature of the data, one of which is sudden short spikes. Taking a similar approach we consider three components which are simulated in order: 
\begin{enumerate}
    \item The base heart rate ($X_t$)
    \item An additional noise process ($Y_t$)
    \item A spike process ($Z_t$)
\end{enumerate}
All variables are one-dimensional and in units of bpm.

The base heart rate provides the general shape of the time series and is where different activities, such as interval training, can be included. $X_t$ represents the `true' heart rate of the athlete. We model $X_t$ using an SDE of the form
\begin{equation} \label{eqn: base heart rate SDE}
    dX_t = -V'(X_t,t)\,dt \,+\, \sigma\,dW_t^{(1)}
\end{equation}
where $V(x,t)$ is an asymmetric potential surface, $\sigma$ represents the amount of noise and $dW_t^{(1)}$ is a standard Brownian motion. $V$ may be time dependent, allowing us to simulate complex activities such as modelling interval training by moving the minimum of $V$, whereas $\sigma$ is constant. $V$ is chosen to be asymmetric as athletes' heart rates often remain close to, but do not regularly exceed, a natural ceiling rate during exercise. However, excursions into lower heart rates, in the region 100-160bpm) were not uncommon in our example data and so are penalised less by the asymmetric potential surface. The precise shape of $V$ and values of $\sigma$ are fit by eye to available data as their purpose is to provide a reasonable background process against which we can try to detect spikes. A plot of $V$ can be found in Fig.\ref{fig:potential surface}.

The additional noise process allows us to incorporate heteroscedasticity, as step changes in the level of noise are sometimes observed during an activity. An example of this in real data is shown in Fig.\ref{fig:heteroscedasticity}. 
In some cases, it is the first ten minutes or so of the activity which is noisier than the rest. A possible explanation \cite{cricklesPaper} of this phenomenon is that at the start of an activity a lack of sweat build-up leads to contact failures in heart rate monitors and hence noisy data. As this noise is considered to be an effect of the monitor, not increased volatility in the athlete's heart rate, we include this as a separate independent process.
\begin{figure}[ht]
    \centering
    \includegraphics[width=\linewidth]{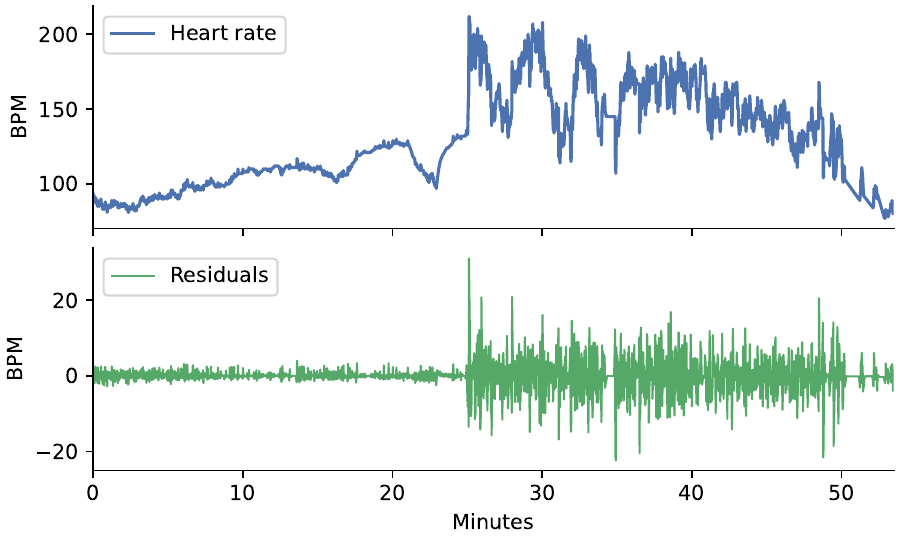}
    \caption{Time series showing heteroscedasticity in a real heart rate sample. Both the recorded heart rate (top) and residuals after removing a moving average with a ten second window (bottom) are shown. It can clearly be seen that there is a significant increase in point variance after about 25 minutes.}
    \label{fig:heteroscedasticity}
\end{figure}

The additional noise process is a simple mean-reverting Ornstein–Uhlenbeck process \cite{gardiner1985handbook}, modelled using the SDE
\begin{equation}
    dY_t = -\alpha(t)\, Y_t\,+\,\beta(t)\,dW_t^{(2)}
\end{equation}
where $\alpha>0$ describes the rate at which the process reverts to its mean at $0$, $\beta$ is a diffusion coefficient and $dW_t^{(2)}$ is a standard Brownian motion, independent of $dW_t^{(1)}$. $\alpha$ and $\beta$ are time-varying but do not depend on the state of the process, and are chosen to create additional noise at the start of the activity.

Finally, the spike process, $Z_t$, is simulated last. As the high heart rates observed during spikes are not considered to be a true reflection of the individual's heart rate \cite{cricklesPaper} it is reasonable to consider them as a separate process. Since arrhythmia-related spikes are also considered to be more likely when the heart rate is high and the heart is under strain \cite{cricklesPaper} $Z_t$ will depend upon $X_t$.

The spike process is given by the increments of a non-homogeneous jump process $J_t$. Let $r:\mathds{R}\rightarrow\mathds{R}$ be a function that describes the rate of spikes as a function of the heart rate. This function is then used to calculate the Poisson intensity
\begin{equation}
    \lambda(t) = r(X_t)
\end{equation}
which is used to simulate a non-homogeneous Poisson process $P_t$ that determines the timing of spikes. The sizes of spikes are then drawn from a log-normal distribution to give the jump process $J_t$. The increments of this jump process give the spikes $Y_t$, hence $Y_t$ is zero except when it is positive at the moment of a spike. 

The final output of the simulations is the synthetic heart rate $(H_t)$ which is found by smoothing $X_t+Y_t+Z_t$ using a moving average, then rounding the result to the nearest integer. The simulation method was examined by Crickles and found to produce realistic heart rate time series. 

We test the spike detection methods described on simulated data with various features including: interval training, step changes in the level of activity and heteroscedasticity, as well as simple situations with a constant level of activity and noise. Example simulated time series can be found in \ref{SI: example simulations}. For each time series the simulation algorithm also outputs a list of spikes containing the time and height of the spikes, as well as the underlying heart rate at the time the spike occurred. We compare this against the spikes detected by each method. 

\subsection{Comparison metrics}

To compare the performance on simulated data of the various spike detection methods described we consider two error metrics. Both use a list of true spikes taken from the simulated data and a list of detected spikes outputted by the chosen spike detection method. 

For the first error metric we set a strict time-window of 5 seconds and scan through the detected spikes. For each detected spike we count it as a `correct spike' if it can be paired with a true spike within the time-window. If so, both spikes are removed from their respective lists. After scanning through the detected spikes we calculate the precision and recall
\begin{align*}
    \text{precision} = \frac{\text{\# correct spikes}}{ \text{\# detected spikes}  } \\
    \text{recall} = \frac{\text{\# correct spikes}}{ \text{\# true spikes}  } 
\end{align*}
and from this the $F_1$ score for each activity
\begin{align*}
    F_1 \text{ score} = \frac{2}{\text{precision}^{-1} + \text{recall}^{-1}}
\end{align*}
which is then averaged over all simulated activities. The $F_1$ score is the harmonic mean of the recall and precision and is a well-established performance metric \cite{dice1945measures,sorensen1948method,chinchor1993muc}. The score takes values between 0 (worst score) and 1 (best score). 

The second error metric, which we call `spike density error' involves transforming the locations of spikes from individual points to a distribution of possible points. The motivation for this is that, although spikes peak at a single point, they do have some width, partially due to the smoothing done in real and simulated heart rate time series. For this reason there is not a single time at which the whole spike event occurs, so it seems reasonable to allow the true and detected spike times to be represented by a distribution rather than a point. This metric also allows us to incorporate the height of a spike. 

We define a kernel density $f:\mathds{R}\rightarrow\mathds{R}$, chosen here to be a Gaussian distribution centered at $0$ with standard deviation $5$ (to correspond to the earlier strict threshold). This kernel density will be translated and scaled to replace the single point value for each spike. Any alternative kernel is acceptable provided it satisfies the following conditions:
\begin{enumerate}
    \item $f$ is positive and symmetric.
    \item $f$ achieves its global maximum at $0$.    
    \item $f$ is integrable and integrates to $1$. 
\end{enumerate}
Denote by $f_s (t)$ the translated and re-normalised function
\begin{equation} \label{eqn: scaled kernel density}
    f_s (t) = \frac{f(t-s)}{\int_0^T f(\tau-s)\,d\tau}
\end{equation}
where $T$ is the length of the time series. 

Given $n$ spikes occurring at times $0\leq s_1<\dots<s_n\leq T$ with corresponding heights $h_1,\dots,h_n$, define the spike density $\rho:[0,T]\rightarrow\mathds{R^{+}}$ by
\begin{equation}
    \rho(t) = \sum_{i=1}^n h_i \, f_{s_i}(t)
\end{equation}

Given lists of real and detected spike times, denote the respective densities by $\rho_r$ and $\rho_d$. The spike density error is given by
\begin{equation} \label{eqn: spike density error}
    \epsilon = \frac{1}{T} \int_0^T \big|\rho_r(t) - \rho_d(t) \big|\,dt
\end{equation}
which provides an error in units of bmp/hour. Properties of the spike density error are discussed in \ref{SI: Spike density error}. We also compare this error against that of the null spike detection model, in which no spikes are ever detected. The error for the null model is simply the integral of the true spike density, so we define the error impact
\begin{align}
    \epsilon^{*} 
    &= \frac{1}{T}\int_0^T \rho_r(t) - \big|\rho_r(t) - \rho_d(t) \big| \,dt 
\end{align}
which again has units bpm/hour. If $\epsilon^{*}$ is negative this means that overall spikes have been detected erroneously and it would have been more accurate to detect no spikes. If $\epsilon^{*}$ is positive this means the spike detection method has reduced the error (i.e. had a positive impact) by correctly identifying spikes, hence a large $\epsilon^{*}$ is indicative of successful spike detection. 

\subsection{Results of comparison}

To compare the spike detection methods that were described in subsection \ref{sec: methods}.\ref{subsection:spike_detection_methods}, 45 training sets of heart rate time series of different activity types (as listed in the final paragraph of subsection \ref{sec: methods}.\ref{sec: simulation algorithm}) were simulated. This training data was used to tune all the relevant parameters in the spike detection methods, for example the width for moving average and the constant multipliers in the thresholding algorithms. A grid search of parameters was done such that each spike detection method maximizes the F1 score and the error impact with the chosen parameter.

Once all the parameters were tuned such that all spike detection methods were achieving their best scores on the set of training data, they were then applied to the test data which contains 900 simulated heart rate time series, similarly of various activity types. From the results of the spike detection, the average F1 scores and average error impact across all simulated activities were evaluated, as shown in Fig.\ref{fig:average F1 score and error impact of all activities}. Furthermore, the average F1 score and average error impact for each subset of simulated activity types were also evaluated for a more detailed analysis.

\begin{figure}[ht]
    \centering
    \includegraphics[width=\linewidth]{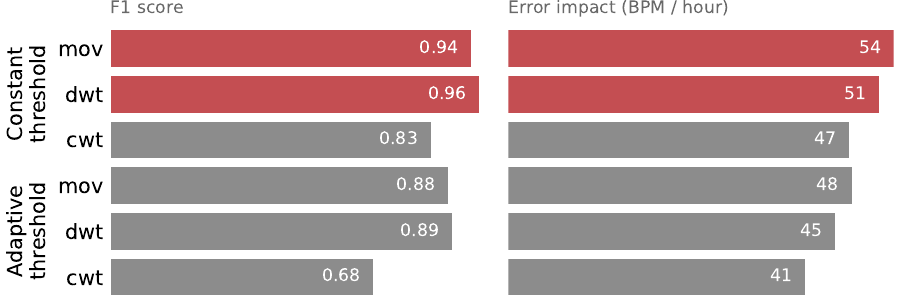}
    \caption{
        Bar charts comparing the F1 score (left) and error impact (right) when averaged across all simulated activities. Both metrics show the methods based on moving average (\texttt{mov}) and DWT smoothing with a constant threshold to be the best performing (both highlighted in red). Of the two, the moving average method was chosen for spike detection on empirical data because of its simplicity.
    }
    \label{fig:average F1 score and error impact of all activities}
\end{figure}

\begin{figure*}[ht!]
    \centering
    \includegraphics[width=.6\linewidth]{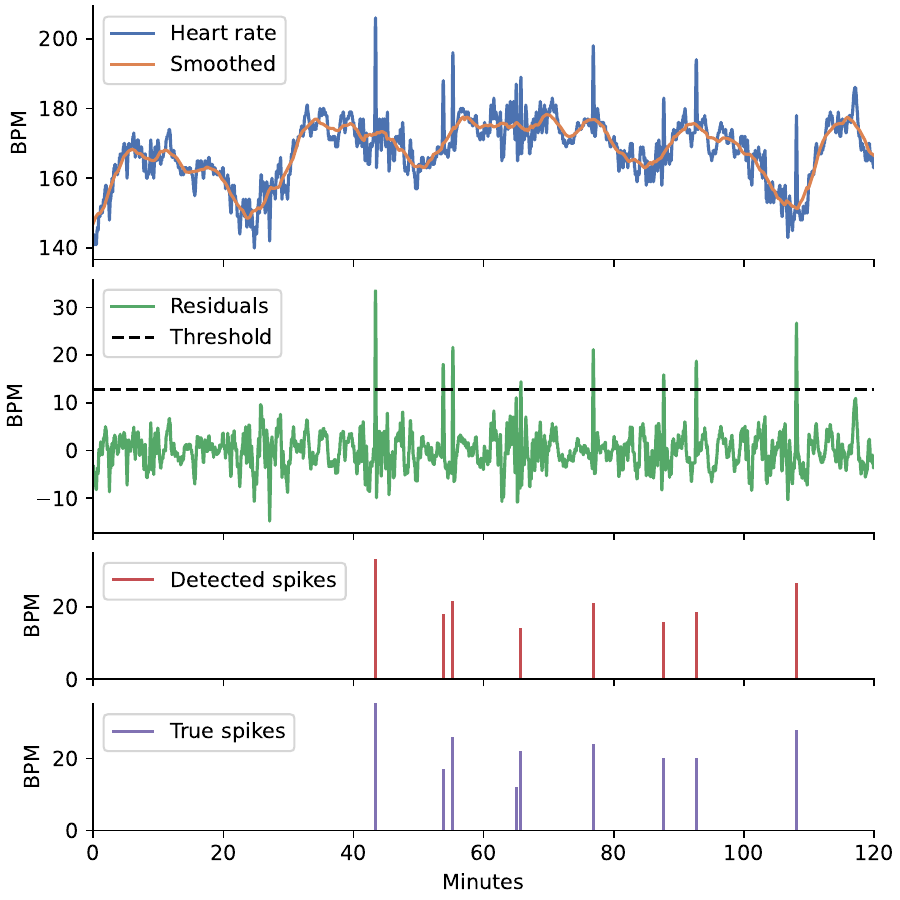}
    \caption{Example of spike detection on simulated heart rate data with method mov\_constant from Fig.\ref{fig:average F1 score and error impact of all activities}. The top plot shows the simulated time series being smoothed by moving average (orange). The second plot shows the resulting residuals (green) and the evaluated constant threshold (dashed). The third and fourth plots are the resulting detected spikes and true simulated spikes respectively, which shows the accuracy of this method in recovering the true spikes.}
    \label{fig:spike detection with moving average and constant threshold example}
\end{figure*}

From Fig.\ref{fig:average F1 score and error impact of all activities}, both moving average and DWT smoothing methods with constant threshold were able to detect the true simulated spikes accurately regardless of activity type, showing similar results in both performance metrics. As the moving average method was able to detect a more accurate spike height, this resulted in a greater error impact. This is illustrated in Fig.\ref{fig:spike detection with moving average and constant threshold example}, which shows a breakdown of this spike detection method on an example simulated heart rate time series. From these results, and because of its relative simplicity, the method with moving average and constant threshold was chosen as the spike detection method to be used going forward.

This spike detection method, with the same set of parameters as used for the comparison, i.e., width $w=200\,\mathrm{s}$ for the moving average smoothing and constant multiplier $2.5$ for thresholding, were then applied to the empirical data of 168 athletes provided by Crickles.

\subsection{Inference} \label{sec: inference}

Once this method has been applied to the empirical data it remains to characterise the frequency of spikes for each athlete. We consider two key lines of enquiry:
\begin{enumerate}
    \item Is there a correlation between the overall frequency of spikes and reports of heart rhythm problems?
    \item Is there a correlation between the frequency of spikes and the heart rate at which they occur? Furthermore, is this relationship the same for individuals with and without heart rhythm problems?
\end{enumerate}

To address the first question we calculate the rate of spike events. We denote the total number of spikes across all an individual's time series by $N$, and the total length of time series by $T$. Define the spike frequency
\begin{equation}
    \lambda = \frac{N}{T}
\end{equation}
which is in fact the maximum likelihood estimate (MLE) if the cumulative number of spike events were a Poisson process with constant rate $\lambda$ (see \ref{SI: Poisson}). This value is calculated for all individuals. As in \cite{cricklesPaper} we then perform a point-biserial correlation check and logistic regression to explore the association between spike rate and heart rhythm problems. 

To address our second line of enquiry, the relationship between spike frequency and heart rate, we consider two approaches. Firstly we consider a more detailed model of the spike frequency. We again assume spikes occur at times determined by a Poisson process, but now allow this process to be in-homogeneous with time-dependent rate $\lambda(t)$. Letting $\overline{X}(t)$ denote a smoothed heart rate time series with spikes removed, we assume $\lambda(t)$ is a piecewise constant function of $\overline{X}(t)$. Given these two models we consider the hypotheses:
\begin{itemize}
    \item $H_0$: spikes occur according to a Poisson process with constant intensity.
    
    \item $H_1$: spikes occur according to a Poisson process whose intensity is a piecewise constant function of the heart rate. 
    
\end{itemize}
We perform a likelihood ratio test to determine if the more complex, heart rate dependent model is a significantly better fit to the observed spike locations. The full details of this hypothesis test are described in \ref{SI: Poisson}. This hypothesis test is performed for each individual and for the athletes collectively. If the null hypothesis were rejected this would indicate there exists a relationship between heart rate and spike rate that requires further investigation. It is also possible that a relationship between inhomogeneity and arrhythmia may be observed. 

We also considered fitting the rate of an inhomogeneous Poisson process whose rate is a continuous function of time or heart rate, but the number of spikes observed in the sample data was too low for such a method to produce meaningful results. 

\section{Results and Discussion} \label{sec: results and discussions}

Spike detection and inference were performed on data from 168 athletes, all of whom had at least 100 activities of different lengths recorded, while 45 of them reported having heart rhythm problems in the Crickles survey. Activities shorter than 5 minutes and those marked as `sticky' were removed from the dataset, leaving a total of approximately 125,000 activities. On average there were 814 hours of activity data for each athlete. 

\subsection{Initial results} \label{sec: initial results}

Fig.\ref{fig:initial box plot} shows on the left a box plot of the Poisson rate of spikes for each athlete, split into those athletes with and without arrhythmia. On the right is a corresponding box plot for the percentage of activates marked as irregular, called the irregular ratio, which is the original measure of gappiness in \cite{cricklesPaper}.

\begin{figure}[ht]
    \centering
    \begin{subfigure}[b]{0.49\linewidth}
        \centering
        \includegraphics[width=\linewidth]{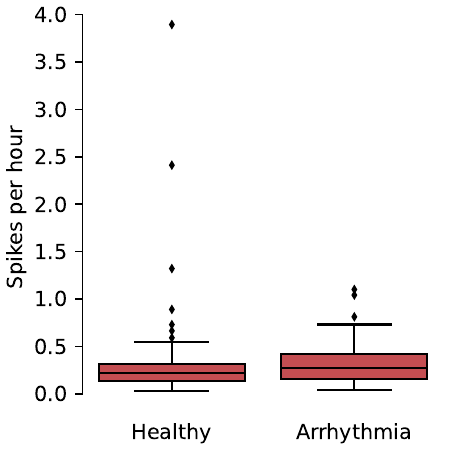}
        \caption{Poisson spike rate}
        \label{fig:initial box plot - poisson rate}
    \end{subfigure}
    \begin{subfigure}[b]{0.49\linewidth}
        \centering
        \includegraphics[width=\linewidth]{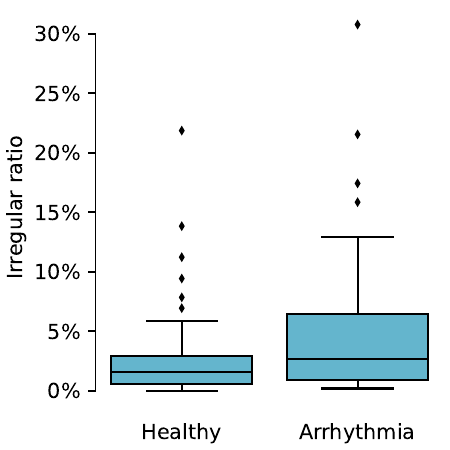}
        \caption{Irregular ratio}
        \label{fig:initial box plot - irregular ratio}
    \end{subfigure}
    \caption{These box plots show the distribution of Poisson spike rates (\subref{fig:initial box plot - poisson rate}) and irregular ratios (\subref{fig:initial box plot - irregular ratio}) for athletes with and without arrhythmia. Some outliers for both Poisson spike rates and irregular ratios can clearly be observed and will be analysed in subsection \ref{sec: results and discussions}.\ref{sec: analysis of outliers}. A correlation with arrhythmia was found from the irregularity ratio, but not for the Poisson spike rates.}
    \label{fig:initial box plot}
\end{figure}

The majority of athletes have a very low rate of spikes, typically fewer than 0.5 spikes per hour. A point-biserial correlation check is performed to asses the correlation between Poisson spike rate and arrhythmia. We also show results of the same test performed on the irregular ratio. 
\begin{table}[H]
    \centering
    \caption{Results of point-biserial correlation check for Poisson spike rate and irregular ratio against cardiac health survey.}
    \begin{tabular}{|c|c|c|}
        \cline{2-3}
        \multicolumn{1}{c|}{} & Correlation & p-value \\
        \hline
        Poisson spike rate & 0.06 & 0.44 \\
        \hline
        Irregular ratio & 0.28 & 0.00027 \\
        \hline
    \end{tabular}
    \label{tab:my_label}
\end{table}
From these values, we cannot conclude a statistically significant correlation between the Poisson spike rate and the reported heart rhythm problems. Conversely, the irregular ratio does shows a significant correlation. We now address several possible causes of this difference.

\subsection{Analysis of outliers} \label{sec: analysis of outliers}

In Fig.\ref{fig:initial box plot - poisson rate} we observe a small number of individuals without arrhythmia for whom the Poisson spike rate is substantially higher than expected. For the purpose of investigating the effect they have on the correlation, outlying athletes can be removed from the dataset by removing those above a threshold percentile of the data. To ensure we compare the Poisson spike rates and irregularity ratios of the same set of athletes we require that athletes are not outliers in either measure. The effect on correlation and p-values of removing outliers is shown in Fig.\ref{fig:effect of outliers}.

\begin{figure}[ht]
    \centering
    \includegraphics[width=\linewidth]{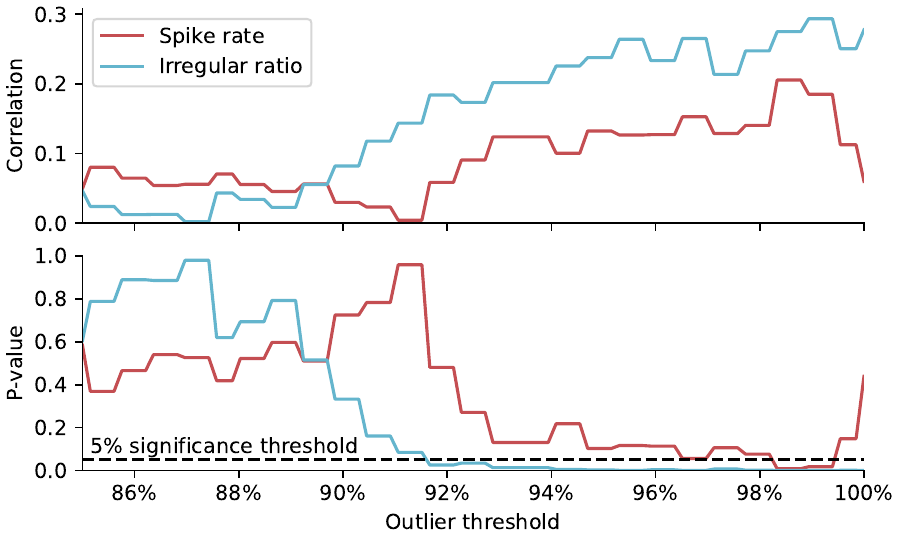}
    \caption{The plots show the correlation coefficient (top) and corresponding p-values (bottom) of Poisson spike rate (red) and irregular ratio (blue) with arrhythmia, against the outlier threshold, where we remove datasets with either of Poisson spike rate or irregular ratio above the threshold percentile of the data. When those above the 98th percentile are removed, the Poisson spike rate shows a significant correlation, indicating that it is sensitive to some outliers present in the data.}
    \label{fig:effect of outliers}
\end{figure}

The correlation between the irregularity ratio and reported heart rhythm problems gradually weakens, with an increasing p-value, as more athletes are removed from the dataset. When considering the Poisson spike rate we observe that removing a small number of outlying athletes (those above the 98th percentile) strengthens the correlation and decreases the p-value, indicating a significant result. However, removing additional athletes then reduces the correlation and significance. This indicates that the correlation between Poisson spike rate and reported heart rhythm problems is highly sensitive to the presence of these outliers. 

Unfortunately, due to our limited access to the data, we were not able to identify the cause of these outliers. It is possible that these athletes' activity time series contain some feature that was not represented in our samples that caused our spike detection method to behave unexpectedly. As the Poisson rate is calculated using spikes across all an athletes activities, errors arising in one activity have the capacity to corrupt the overall rate. This sensitivity to individual activities may contribute to these outlying results. By comparison the Boolean classification of `irregular' activities for the gappiness metric in \cite{cricklesPaper} limits the influence of any one activity on the athlete's irregular ratio, leading to greater robustness at the cost of reduced information. 

Alternatively it may be that the spike detection method has performed as expected, but that outliers have a higher rate of spikes due to a fault in their heart rate monitors or other non-arrhythmia related effect. The precise mechanism by which irregular heart rhythms translates to gappy time series is unclear, and the spikes we observe have other possible causes. Hence, even if the spike detection methods operate correctly, the inferred rate of spikes still may not correlate clearly with arrhythmia. 

Finally it may possibly be the case that outlying individuals have an un-diagnosed heart condition, meaning that the survey results do not accurately correspond to true heart rhythm problems. Additionally there may be other individuals who responded as having had heart rhythm problems but who are now receiving treatments. As the methods proposed here are not intended as a diagnostic tool we do not speculate as to the heart health of any of the athletes whose data we analyse.

As the cause of these outliers is unclear we would not feel justified in excluding them from our analysis, and so recognise that the current method of spike detection and inference does not improve upon the original gappiness metric proposed in \cite{cricklesPaper}.

\subsection{Analysis of heart rate dependence} \label{sec: analysis of hr dependence}

For the in-homogeneity hypothesis test of spike rates for all athletes, as explained in section \ref{sec: methods}.\ref{sec: inference}, the results show that for 123 athletes out of the 168, there was significant evidence to reject the null hypothesis, suggesting that spikes occur according to a Poisson process with a rate which is heart rate dependent. A histogram of the frequency of spikes against the heart rate where the spikes occur is shown in Fig.\ref{fig:distribution of spikes in hr}, which shows that a large number of spikes were detected in the range $105-140$bpm.

From the paper \cite{cricklesPaper} by Crickles, it is also hypothesized that only features of heart rate data at high ranges above the lactate threshold heart rate (LTHR) of the athlete are informative of arrhythmia, since the athlete is under exertion at that range of heart rate. Hence, to test our results against this hypothesis, spikes below a certain heart rate threshold $h_{\text{thresh}}$ are removed for all athletes. The corresponding Poisson spike rate for each athlete is then recalculated, and similarly, the point-biserial correlation between the new Poisson spike rate and arrhythmia is evaluated. The plot in Fig.\ref{fig:heart rate dependence correlation plot} shows the correlation coefficients and p-values against increasing heart rate threshold, $h_{\text{thresh}}$.

In the initial section of the plot in the range $h_\text{thresh}\leq 75$, both the correlation coefficient and p-value stay relatively constant at the value prior to filtering out any spikes. This is because there are only a very small number of spikes present in the low range of the heart rate, so the correlation tests produce a similar result to that found in subsection \ref{sec: results and discussions}.\ref{sec: initial results}.

\begin{figure}[ht]
    \centering
    \includegraphics[width=\linewidth]{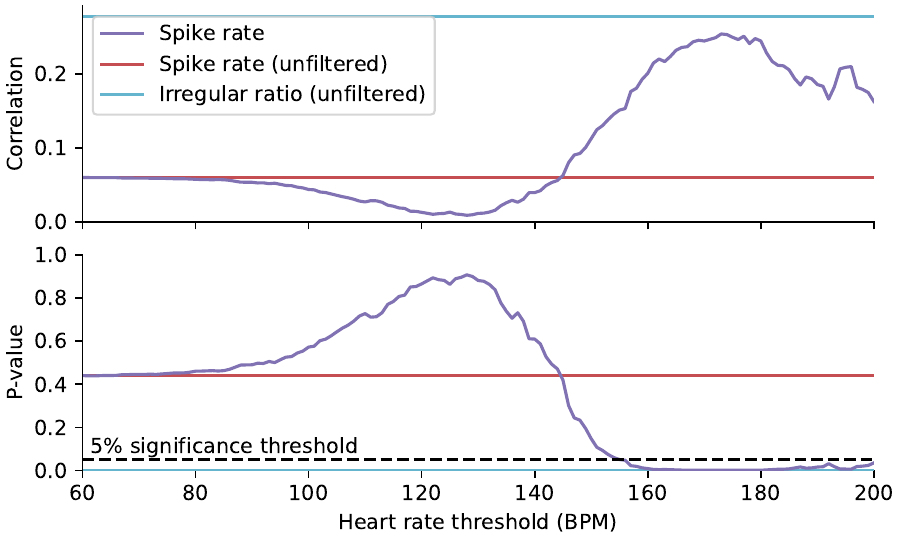}
    \caption{Correlation with arrhythmia against heart rate threshold, where the top figure shows the correlation coefficient and bottom figure are the corresponding p-values. The horizontal blue and red lines display the results of irregularity ratio from Crickles and Poisson spike rate without filtering out any spikes respectively. The purple curve shows the Poisson spike rate with spikes below the heart rate threshold being removed.}
    \label{fig:heart rate dependence correlation plot}
\end{figure}

Within the heart rate threshold range $75\leq h_\text{thresh}\leq 130$, there is a decrease in correlation between Poisson spike rate and arrhythmia, which is caused by it being a gray area where this heart rate range is considered as high for some athletes and low for others. From the histogram in Fig.\ref{fig:distribution of spikes in hr}, the modal value of heart rate where spikes occur also falls in this range. Therefore, past this range, a large number of spikes would be removed relatively randomly from all athletes, so affecting the significance of the spikes detected and thus reducing the correlation with arrhythmia.

Within the range $130\leq h_\text{thresh} \leq 175$, however, a steady increase in correlation coefficient and decrease in p-value is observed. The correlation of spike rate with arrhythmia peaks at $h_\text{thresh}=175$. In this range, the heart rate is likely considered to be high for all athletes. Since now only spikes at high heart rates remain, the increase in correlation justifies that only spikes at a higher heart rate range are significantly indicative of reports of heart rhythm problems. This not only agrees with the hypothesis as mentioned above, but also fits with the medical understanding of arrhythmia that athletes under exertion may induce irregularity patterns in heart rate \cite{exercise_induced_arrhythmia}. The correlation value with arrhythmia using Poisson spike rate never achieves a value as high as using Crickles' irregularity ratio, since only a global heart rate spike threshold was used here, instead of an individualised LTHR for each athlete, as was done in Crickles' research. Hence, further research to obtain LTHR data and implement this into the spike detection methods would likely further improve the correlation of Poisson spike rates with arrhythmia. 

Finally when $h_\text{thresh}\geq 175$, both correlation and p-values decrease as at this point most spikes have been removed. More heart rate data with spikes at such high heart rates would be needed to obtain a better correlation and draw a meaningful conclusion.

\section{Further Work} \label{sec: further work}

The first step of further work on this project would be as mentioned in the second last paragraph of section \ref{sec: analysis of hr dependence}, where the LTHR for each athlete should be incorporated in the spike detection methods as a lower bound so that only spikes at heart rate ranges considered high for each individual athlete are considered. This would require access to further historical activity data from the athlete for an accurate estimate of the athlete's LTHR, which will likely be a challenge due to data privacy.

With the current spike detection method, it is also observed that a large proportion of spikes were detected from the first $10$ minutes of the activity data, shown in the histogram in Fig.\ref{fig:distribution of spike times}. This may be largely caused by the defects of heart rate monitors that struggle to obtain accurate readings when athletes start their activities. For example, chest straps are prone to contact failures with the skin before the build up of sweat, resulting in inaccurate and noisy heart rate readings. More analysis will need to be done on coping with noisy data from the heart rate monitor such as this, so that the spike detection methods can properly distinguish noise caused by external effects from spikes caused by heart rhythm irregularities. To cope with initial noisy readings, one method will be to discard readings from the first $5$ to $10$ minutes of activity before implementing the spike detection methods.

Instead of explicitly detecting spikes in heart rate data, further work can also be done to characterise other features of heart rate time series, and their correlation with arrhythmia. It is observed from empirical data that the volatility of heart rate varies not only between separate activities but also within heart rate data recorded by an athlete from a single activity. Hence, analysing and measuring the volatility of heart rate data could also potentially be informative of heart rhythm problems. 

In addition to heart rate measurements, sport devices often record other data such as the velocity, distance, elevation, power output by the athlete, etc., during the activity. Analysis on this additional data could perhaps be used to filter out any characteristics in the heart rate data such as change in volatility or drastic trends in heart rate that are caused by a sudden exertion from the athlete while exercising.

In conclusion, other than exploring different data driven methods, there is scope for lots of further work to be done in improving the current spike detection method and modelling other features of the of heart rate time series, in order to obtain stronger correlations of these metrics with arrhythmia.

\newpage

\acknow{The authors would like to thank professor Colm Connaughton for his excellent supervision, and Ian Green for his invaluable insights and tireless running of our experiments on the Crickles data. This work was supported by the Engineering and Physical Sciences Research Council through the Mathematics of Systems II Centre for Doctoral Training at the University of Warwick (reference EP/S022244/1).}

{\showacknow{}} 

\section*{Bibliography}
\addcontentsline{toc}{section}{Bibliography}
\bibliography{bibliography.bib}

\newpage
\appendix
\renewcommand{\thesection}{S\arabic{section}}
\renewcommand{\thefigure}{S\arabic{section}.\arabic{figure}}

\phantomsection\addcontentsline{toc}{section}{Supplementary Information}
{\Huge\textbf{Supplementary Information} }

\section{Spike detection methods} \label{SI: spike detection methods}
\setcounter{figure}{0}

\subsection{An introduction to wavelets}
What follows is a very fast paced overview of the relevant theory of wavelets for this project. For a friendlier introduction the reader is referred to the excellent texts by Daubechies \cite{daubechies1992wavelets} and Kaiser \cite{kaiser2011wavelets}. In particular, Daubechies' book formed the basis for this section.

The Fourier transform gives a decomposition of a signal into a frequency spectrum. Low frequency components describe the slower changing aspects of the signal such as the overall trend, while features like spikes require the higher frequency components. However, representation of a spike in Fourier space requires contributions from all frequencies. Indeed, the most extreme case of a spike is the Dirac delta function centered at a point \(c \in \mathbb{R}\), \(\delta_c(t)\), whose Fourier transform is \(\hat{\delta_c}(\omega) = e^{-i \omega c}\). This makes it hard to use the Fourier transform to isolate spikes within a signal and it is therefore desirable to look for a decomposition of the signal which localises `content' in both frequency and time.

Wavelet theory \cite{daubechies1992wavelets, mallat2009wavelets} covers many related transforms to perform this time-frequency decomposition, the most important being the continuous and discrete wavelet transforms. They each begin with a \emph{mother wavelet}, \(\psi \in L^2(\mathbb{R})\), such as the one in Fig.\ref{fig:mother wavelet cts}, which defines a short-lived oscillation. By translating and rescaling \(\psi\), we can generate a family of \emph{(child) wavelets}
\begin{equation}\label{eqn: child wavelet}
    \psi^{a, b}(t) = |a|^{-1/2} \psi\left(\frac{t - b}{a}\right).
\end{equation}
Here \(a \in \mathbb{R}\setminus \{0\}\) is a scaling parameter and \(b \in \mathbb{R}\) is a translation parameter. By rescaling by a factor of \(|a|^{-1/2}\) we preserve the \(L^2\) norm of the wavelet, \(\| \psi^{a, b} \| = \| \psi \|\).

\begin{figure}[htb]
    \centering
    \includegraphics[width=\linewidth]{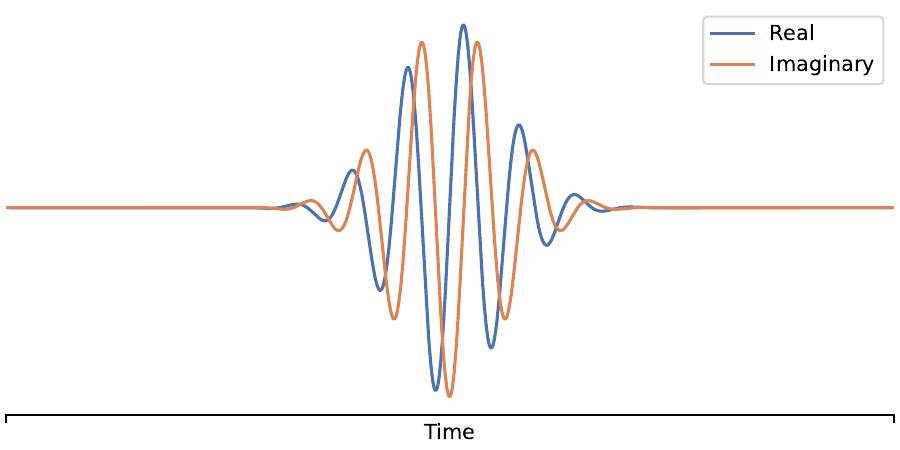}
    \caption{An example of a mother wavelet: the complex Morlet wavelet with bandwidth \(2\,\mathrm{s}^2\) and central frequency \(1 \,\mathrm{Hz}\), \(\psi(t) = \frac{1}{\sqrt{2 \pi}} e^{-t^2/2}e^{2 \pi i t}\).}
    \label{fig:mother wavelet cts}
\end{figure}

We must place some restrictions on what is an acceptable mother wavelet. It is sufficient to require that \(\psi \in L^1(\mathbb{R}) \cap L^2(\mathbb{R})\) with
\begin{equation}\label{eqn: admissible wavelet}
    \int_\mathbb{R} \psi(t) \,dt = 0,\qquad
    \|\psi\| = \int_\mathbb{R} |\psi(t)|^2 \,dt = 1.
\end{equation}
The requirement that \(\psi\) integrates to zero will allow us to invert the wavelet transforms and such a mother wavelet is called \emph{admissible}. The condition on the \(L^2\)-norm is just a normalisation convention. It is possible to weaken these restrictions so that \(\psi\) need not be in \(L^1(\mathbb{R})\) and need not integrate to 0 but is still admissible, however, this is of little practical use.

The \emph{continuous wavelet transform (CWT)} of a signal \(f(t)\) is
\begin{equation}\label{eqn: cwt}
\begin{split}
    (\mathcal{W} f) (a, b) &= \langle f, \psi^{a,b} \rangle \\
    &= |a|^{-1/2} \int_\mathbb{R} f(t) \overline{\psi}\left( \frac{t-b}{a} \right) dt .
\end{split}
\end{equation}
Like the Fourier transform, it is defined using an inner product. However, unlike with the Fourier transform, the CWT includes a lot of redundancy. This means that there is no single formula for inverting the CWT. However, a canonical formula known as the \emph{resolution of identity} does exist
\begin{equation}\label{eqn: cwt inverse}
    f(t) = C_\psi^{-1} \int_{-\infty}^\infty \int_{-\infty}^\infty (\mathcal{W}f)(a, b)\, \psi^{a, b}(t)\, \frac{1}{a^2}\, da\, db
\end{equation}
where
\begin{equation*}
    C_\psi = 2 \pi \int_{-\infty}^\infty |\hat{\psi}(\xi)|^2 \frac{1}{|\xi|} d\xi .
\end{equation*}
Here \(\hat{\psi}\) denotes the Fourier transform of the mother wavelet.
It can be shown that for an admissible wavelet (\eqref{eqn: admissible wavelet}), the integral defining the constant \(C_\psi\) exists and is finite.

In the \emph{discrete wavelet transform (DWT)}, the scaling and translation parameters only take discrete values
\begin{equation} \label{eqn: dwt discretisation}
    a_m = 2^{m},\quad b_n = n 2^m
\end{equation}
giving child wavelets
\begin{equation}
    \psi_{m, n}(t) = \frac{1}{\sqrt{2^m}} \psi\left(\frac{t}{2^m} - n\right).
\end{equation}
Therefore the DWT can be viewed as a discrete sampling of the CWT, as illustrated in Fig.\ref{fig:dwt sampling}. At higher frequencies, the DWT has higher temporal resolution to capture the faster changes but lower frequency resolution since this is more important for lower frequency components.
Note that at this point, we are still in the regime where the original signal is sampled continuously.

\begin{figure}[htb]
    \centering
    \includegraphics[width=\linewidth]{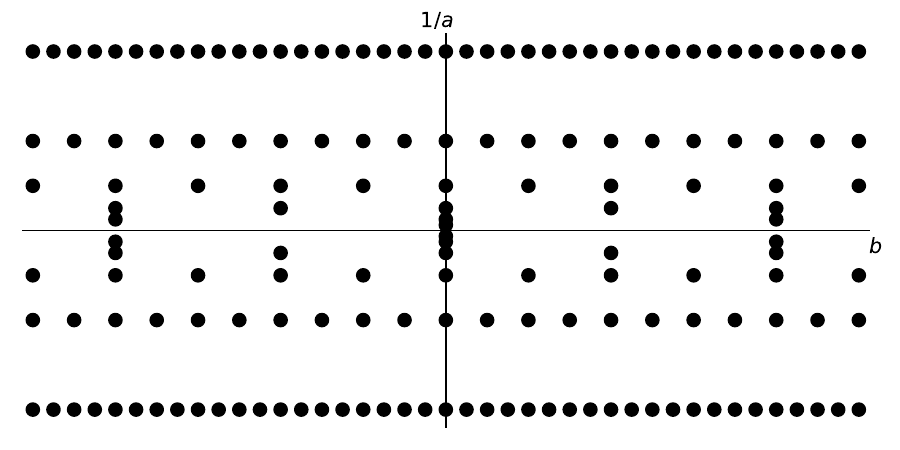}
    \caption{The discrete wavelet transform can be viewed as a discrete sampling of the continuous wavelet transform. The dots in this figure show these discrete locations. Note that the \(y\)-axis shows \(1/a\) which is proportional to the central frequency of the wavelet. The \(x\)-axis plots the translation parameter, \(b\).}
    \label{fig:dwt sampling}
\end{figure}

For particular choices of the mother wavelet, \(\psi\), the wavelets \(\psi_{m,n}\) form a basis for \(L^2(\mathbb{R})\). For more specific choices, this basis is orthonormal and the coefficients in the decomposition can be calculated easily by taking the inner product with the respective wavelets
\begin{equation}
    f(t) = \sum_{m, n} \langle f,  \psi_{m,n} \rangle \psi_{m,n}.
\end{equation}
If the basis is not orthonormal then we can still calculate the coefficients using the Riesz representors, \(\tilde{\psi}_{m,n}\), of the dual basis. These are functions satisfying \( \langle \psi_{m,n}, \tilde{\psi}_{m',n'} \rangle = \delta_{m,m'} \delta_{n,n'} \), where \(\delta_{j,k}\) is the Kronecker delta. The coefficients of \(f\) in the wavelet basis are given by taking the inner product with this dual basis,
\begin{equation}
    f(t) = \sum_{m, n} \langle f,  \tilde{\psi}_{m,n} \rangle \psi_{m,n}.
\end{equation}
The choice of discretisation made in \eqref{eqn: dwt discretisation} is motivated by multiresolution analysis which will be introduced shortly, however other discretisations are possible. The subject resulting from a combination of mother wavelet and discretisation which doesn't yield a basis is that of redundant \emph{frames} and is beyond the scope of this introduction.

For a practical understanding of wavelets, it is necessary to introduce multiresolution analysis. This is by far the most common implementation for the DWT and, for example, is used by \texttt{pywavelets} \cite{lee2019pywavelets}, the library used for this project. Aside from computational efficiency, one main reason for this is that it provides a method to construct wavelets which form an orthonormal basis.

Rather than explicitly defining the mother wavelet, \(\psi\), multiresolution analysis instead begins with a nested sequence of (closed) approximation spaces
\[\cdots \subset V_2 \subset V_1 \subset V_0 \subset V_{-1} \subset V_{-2} \subset \cdots\]
satisfying
\[
\overline{\bigcup_{m \in \mathbb{Z}} V_m} = L^2(\mathbb{R}),\qquad
\bigcap_{m \in \mathbb{Z}} V_m = \{0\}.
\]
We further require that
\begin{align*}
    \forall n \in \mathbb{Z}, \quad f \in V_0 \;&\Leftrightarrow\; t \mapsto f(t - n) \in V_0, \\
    \forall m \in \mathbb{Z}, \quad f \in V_m \;&\Leftrightarrow\; t \mapsto f(2^m t) \in V_0.
\end{align*}
The latter condition here is what gives multiresolution analysis its name.
Finally, we require the existence of a function \(\phi \in V_0\) such that \(\{\phi_{0, n} : n \in \mathbb{Z}\}\) form an orthonormal basis for \(V_0\), where for all \(m, n \in \mathbb{Z}\), \(\phi_{m,n} = 2^{-m/2} \phi(2^{-m}t - n)\). This \(\phi\) is called the \emph{scaling function} for the multiresolution analysis.

With this construction, for \(f \in L^2(\mathbb{R})\) we can build successive approximations \(f_m \in V_m\) by projecting orthogonally onto \(V_m\). By construction, we have \(f_m \to f\) pointwise as \(m \to -\infty\).
Furthermore, it is clear that the scaling function, \(\phi\), is sufficient to reconstruct the full multiresolution analysis set-up.

We are now in a position to construct an orthonormal wavelet basis for \(L^2(\mathbb{R})\).
For each \(m \in \mathbb{Z}\), let \(W_m\) be the orthogonal complement of \(V_m\) in \(V_{m-1}\), so that
\[V_{m-1} = V_m \oplus W_m, \quad\text{and}\quad W_m \perp W_{m'} \;\text{whenever}\; m \neq m'. \]
Define
\begin{equation}
    \psi = \sum_{n \in \mathbb{Z}} g_n \phi_{-1, n} \;\in W_0 \subset V_{-1}
\end{equation}
where for each \(n \in \mathbb{Z}\), \(g_n = (-1)^n \overline{\langle \phi, \phi_{-1, -n+1} \rangle}\).
It is proved in \cite[Theorem 5.1.1]{daubechies1992wavelets} that \(\psi\) is a mother wavelet generating an orthonormal basis of wavelets \(\{\psi_{m,n} : m, n \in \mathbb{Z}\}\) for \(L^2(\mathbb{R})\) where each \(\psi_{m, n}(t) = 2^{-m/2}\psi(2^{-m}t - n)\). Further, for each \(m \in \mathbb{Z}\), \(\{\psi_{m, n} : n \in \mathbb{Z}\}\) are an orthonormal basis for the subspace \(W_m\).
Note however, that this choice of \(\psi\) is not unique.

We now describe an efficient algorithm for calculating these coefficients, which is an example of something known as a subband filtering scheme.
Suppose that we have a fine scale approximation \(f \in V_m\) of some function \(f_*\) given by orthogonal projection, and that we have calculated the coefficients \((\langle f, \phi_{m, n} \rangle \,: n \in \mathbb{Z})\) in the basis generated by the scaling function. By rescaling our scaling function, we may assume without loss of generality that \(m=0\).
Thus we have \(f = \sum_{n \in \mathbb{Z}} \langle f, \phi_{0, n} \rangle \phi_{0, n} \in V_0\).
It is our aim to replace this representation of \(f\) with one using the natural basis for \(V_0 = V_M \oplus W_M \oplus \cdots \oplus W_1\) for a given integer \(M \geq 1\) called the \emph{level} of the decomposition,
\begin{equation}
    f = \sum_{n \in \mathbb{Z}} a_{M, n} \phi_{M, n} + \sum_{m = 1}^M \sum_{n \in \mathbb{Z}} d_{m, n} \psi_{m, n}.
\end{equation}
The coefficients \(a_{M, n}\) on the scaling function part of the basis will be called the \emph{approximation coefficients}, while the coefficients \(d_{m, n}\) on the wavelet part of the basis will be called the \emph{detail coefficients}.

We have already defined coefficients \(g_n\) and can define coefficients \(h_n\) such that
\begin{equation}
    g_n = \langle \psi, \phi_{-1, n} \rangle,\quad
    h_n = \langle \phi, \phi_{-1, n} \rangle.
\end{equation}
It is shown in \cite[section 5.6]{daubechies1992wavelets} that for \(f \in L^2(\mathbb{R})\) and for all \(m, n \in \mathbb{Z}\),
\begin{subequations}
\label{eqn: mra recursion}
\begin{align}
    d_{m,n} = \langle f, \psi_{m,n} \rangle &= \sum_{k \in \mathbb{Z}} \overline{g_{k-2n}} \langle f, \phi_{m-1, k} \rangle, \label{eqn: mra detail recursion} \\
    a_{m,n} = \langle f, \phi_{m,n} \rangle &= \sum_{k \in \mathbb{Z}} \overline{h_{k-2n}} \langle f, \phi_{m-1, k} \rangle. \label{eqn: mra approx recursion}
\end{align}
\end{subequations}
This gives us a recursive algorithm for computing the approximation and detail coefficients of \(f\) up to any level \(M \geq 0\), beginning with the inner products we have already computed \(\{\langle f, \phi_{0, n} \rangle : n \in \mathbb{Z}\}\). Indeed, using \eqref{eqn: mra recursion} we can generate the coefficients for level \(m+1\) by convolving the approximation coefficients for level \(m\) with the respective filters \((\overline{g_{-n}})\) and \((\overline{h_{-n}})\) and discarding the odd terms in the result. This subsampling of the coefficients is referred to as \emph{decimation}. The pair of filters \((\overline{g_{-n}})\) and \((\overline{h_{-n}})\) are known in the signal processing literature as a \emph{filter bank}.

Of course, in practice we do not have a continuously sampled signal. Instead, we assume sampling at regular intervals. This actually works to our advantage. Firstly, we can perform all the convolutions very efficiently using a fast Fourier transform. Secondly, we do not need to further approximate \(f\) at the start of the multiresolution analysis. Rather, we can exactly compute the first set of approximation coefficients \(a_{0, k} = \langle f, \phi_{0, k} \rangle\) for \(k \in \mathbb{Z}\) using a convolution. In this setting, the original sampling of \(f\) can be exactly recovered no matter what level \(M\) of decomposition is used. However, length of the time series limits how large we can make \(M\). Indeed, at each step the coefficients are decimated by a factor of two and this cannot be continued forever.

\begin{figure*}[p]
    \centering
    \begin{subfigure}[b]{0.45\linewidth}
        \centering
        \includegraphics[width=\linewidth]{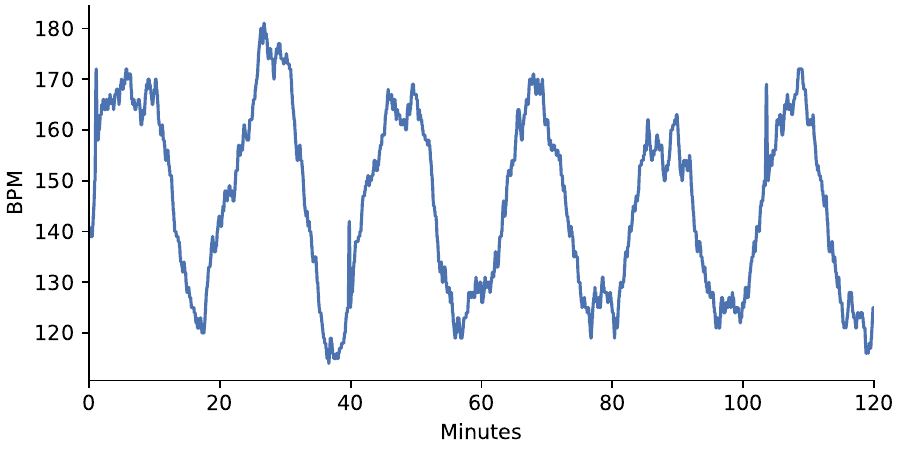}
        \caption{Original time series}
        \label{fig:dwt decomposition - original}
    \end{subfigure}
    \hspace{3em}
    \begin{subfigure}[b]{0.45\linewidth}
        \centering
        \includegraphics[width=\linewidth]{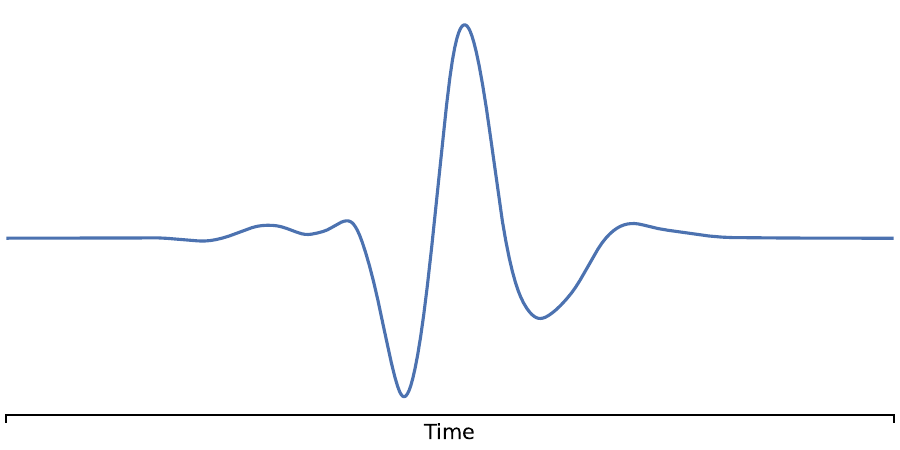}
        \caption{The symlet wavelet of order five}
        \label{fig:dwt decomposition - wavelet}
    \end{subfigure}
    \par\bigskip  
    \begin{subfigure}{\linewidth}
        \centering
        \includegraphics[width=\linewidth]{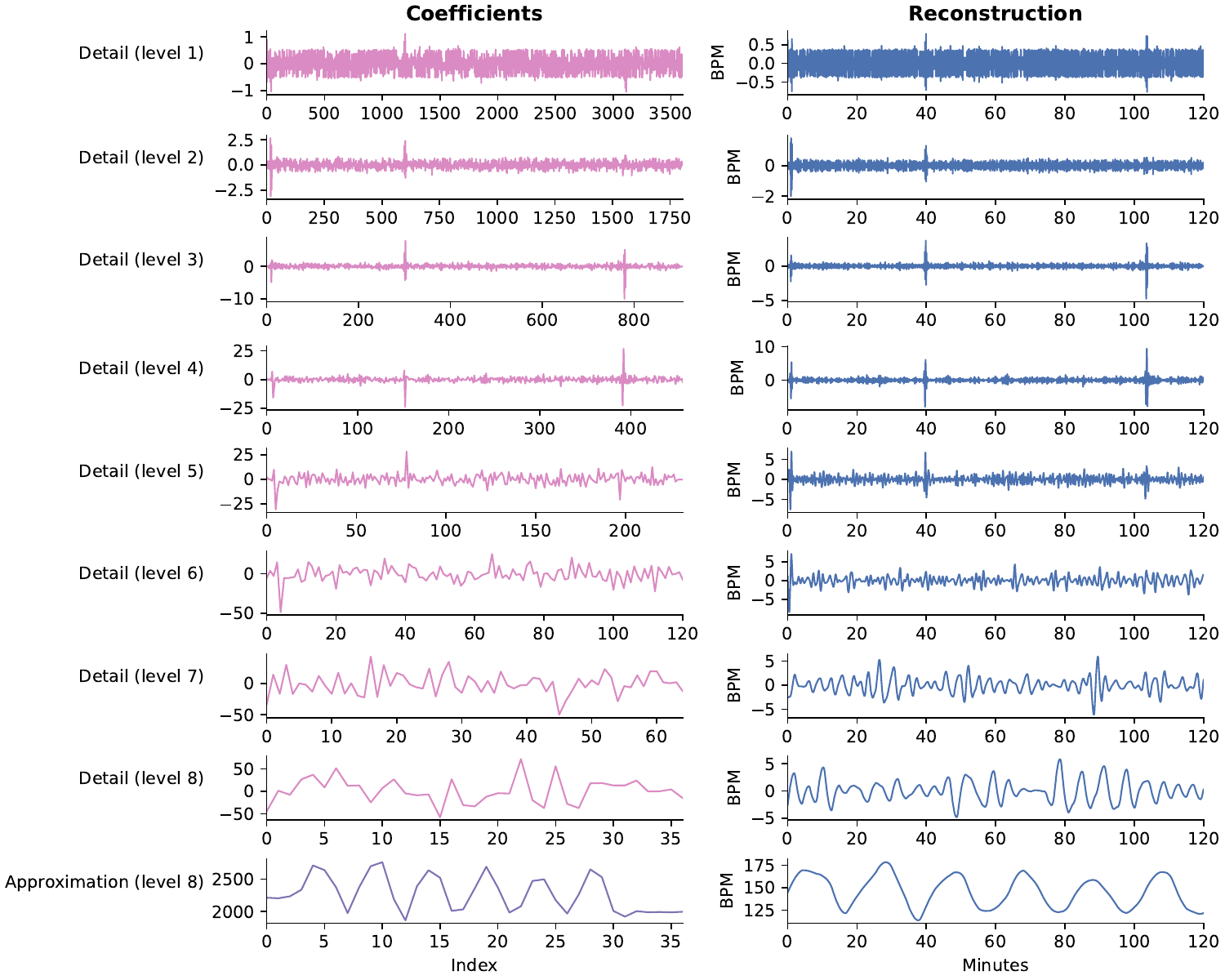}
        \caption{DWT decomposition}
        \label{fig:dwt decomposition - decomposition}
    \end{subfigure}
    \caption{
        An example of a multiresolution decomposition using DWT.
        Figure (\subref{fig:dwt decomposition - original}) shows the original, simulated time series; figure(\subref{fig:dwt decomposition - wavelet}) shows the fifth order symlet wavelet implicitly defined by the multiresolution analysis and figure (\subref{fig:dwt decomposition - decomposition}) shows the detail and approximation coefficients for the first eight decomposition levels. The decomposition was performed after extending the signal symmetrically at both ends to reduce end effects in the coefficients. In figure (\subref{fig:dwt decomposition - decomposition}), the coefficients (left) are decimated at each level so that each successive level has half as many coefficients. Their amplitude grows with the level of the decomposition because the wavelet basis is scaled down by a factor of \(\sqrt{2}\) at each level. The reconstructions (right) show the contribution of each level to the original time series and summing these will recover the original time series.
    }
    \label{fig:dwt decomposition}
\end{figure*}

More generally, when we have a discretely sampled signal we must be careful not to trust wavelet coefficients corresponding to frequencies higher than the Nyquist frequency. This is especially important for the CWT, where the algorithm does not naturally prevent us from going beyond this. Unlike in Fourier analysis however, this is just a rule of thumb, since the concept of a single frequency is stronger for some families of mother wavelet than for others.

It is also common in a practical setting to extend the signal in some way to avoid end effects. By default we could just interpret the signal as zero outside the sampled interval. However, by extending it symmetrically, we can reduce end effects in the coefficients for each level.

Figure \ref{fig:dwt decomposition} exhibits the decomposition of a simulated heart rate using an orthogonal wavelet called the `symlet' of order 5.

Finally, we note that we can relax the condition on the basis \(\{\phi_{-1, n} : n \in \mathbb{Z}\}\) being orthonormal \cite[section 7.4]{mallat2009wavelets}. If instead they simply form a Riesz basis for \(V_0\) then the wavelet basis we obtain for \(L^2(\mathbb{R})\) is not orthonormal, but coefficients in this basis can be calculated using the dual basis, as before. Such wavelet bases are often called \emph{biorthogonal}. The filter banks for the corresponding multiresolution analysis contain four filters: one pair for analysis (deconstruction) and a second pair for synthesis (reconstruction).

\subsection{Spike detection using the continuous wavelet transform}
We can use the continuous wavelet transform to perform spike detection on our heart rate time series. As can be seen in figure \ref{fig:spike detection with cwt}, spikes manifest as ridges into the smaller scales (higher frequencies). Choosing an appropriate scale, we obtain a function of the translation parameter alone which we can threshold to classify spikes. As with the methods which threshold residuals of the heart rate time series, we can use either a constant threshold or an adaptive threshold here.

The choice of wavelet is important for performing spike detection using the CWT. A real-valued wavelet will give an oscillating pattern for each spike, making it difficult to know how many spikes there are and where they are located. Therefore, in this report we use a complex wavelet from the family of complex Morlet (or Gabor) wavelets. These have form
\begin{equation} \label{eqn: complex morlet wavelet}
    \psi(t) = \frac{1}{\sqrt{b \pi}} e^{-t^2/b}e^{2 \pi i c t}
\end{equation}
where the bandwidth, \(b\), and central frequency, \(c\), are parameters. An example is shown in figure \ref{fig:mother wavelet cts}. In this project, we used values of \(b = 1\,\mathrm{s}^2\) and \(c = 1 \,\mathrm{Hz}\).
The benefit of using a complex wavelet is that instead of oscillating either side of zero, the wavelet coefficients instead orbit zero in the complex plane. Therefore, taking the absolute value of the coefficients yields a single smooth hump for each spike.

After thresholding the CWT coefficients, we are left with several contiguous regions in which a spike occurs. We assign a spike to each of these regions by taking the midpoint.

\begin{figure}[htb]
    \centering
    \includegraphics[width=\linewidth]{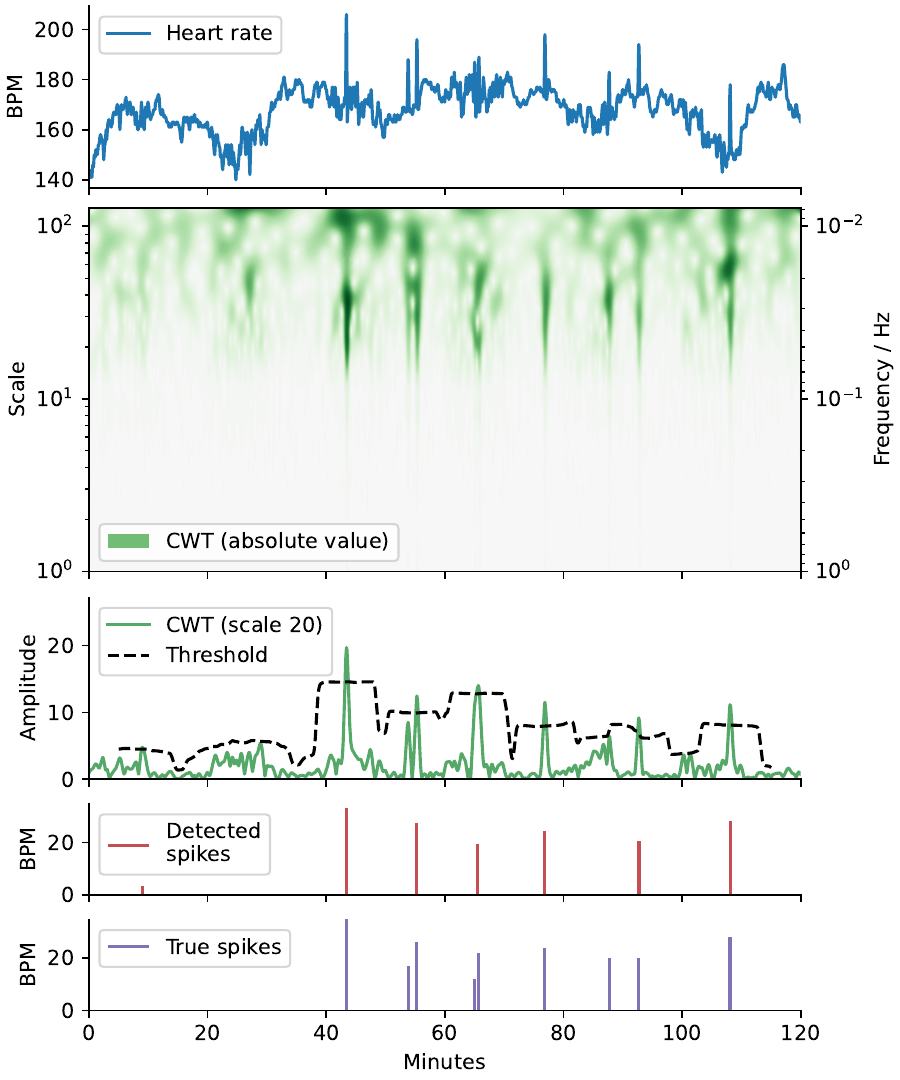}
    \caption{
        Example of spike detection with a continuous wavelet transform. First the simulated heart rate in the top plot is passed through a continuous wavelet transform using a complex Morlet wavelet with bandwidth \(2\,\mathrm{s}^2\) and central frequency \(1\,\mathrm{Hz}\) (see Fig.\ref{fig:mother wavelet cts}). Its absolute values are shown in the second plot. An adaptive threshold is then used on the CWT at scale \(20\,\mathrm{s}\) (\(0.05\,\mathrm{Hz}\)), shown in the third plot. The bottom two plots show that most of the simulated spikes were detected in this example, although three were missed. Furthermore, an additional spike was detected at the beginning of the time series, although this was assigned a very small height so will not greatly affect the error impact. 
    }
    \label{fig:spike detection with cwt}
\end{figure}

\subsection{Spike detection using the discrete wavelet transform} \label{SI: spike detection - dwt}
In the main text we described our principal method of spike detection, which proceeded by
\begin{enumerate}
    \item smoothing the original signal,
    \item subtracting the smoothed signal from the original,
    \item using a threshold (constant or adaptive) to identify spikes in the residuals.
\end{enumerate}
The discrete wavelet transform provides an alternative smoothing method for this process.

In \emph{signal denoising} \cite[chapter 11]{mallat2009wavelets}, we attempt to find a sparse representation for our signal. As we have already seen, using multiresolution analysis we can obtain a representation using a set of approximation and detail coefficients. As is visible in figure \ref{fig:dwt decomposition - decomposition}, not all of these detail coefficients are significant, and many could reasonably be regarded as noise. By shrinking the smaller coefficients to zero, we can obtain a sparse approximation to the original time series.

There are a few choices for the type of thresholding used to obtain the sparse representation. The most popular are the \emph{hard threshold}, \(T_\lambda(x)\), and the \emph{soft threshold}, \(S_\lambda(x)\), since they are most amenable to statistical risk analysis. However, in this project we used the \emph{non-negative garotte} \cite{gao1998garotte}, \(G_\lambda(x)\), which has desirable features from both. The three methods are given by
\begin{subequations}
\begin{allowdisplaybreaks}
\begin{align}
    T_\lambda(x) &= \begin{cases}
        x & \text{if } |x| \geq \lambda, \\
        0 & \text{if } |x| < \lambda;
    \end{cases} \\
    S_\lambda(x) &= \begin{cases}
        x - \lambda & \text{if } x \geq \lambda, \\
        x + \lambda & \text{if } x \leq -\lambda, \\
        0 & \text{if } -\lambda < x < \lambda;
    \end{cases} \\
    G_\lambda(x) &= \begin{cases}
        x - \lambda^2/x & \text{if } |x| \geq \lambda, \\
        0 & \text{if } |x| < \lambda.
    \end{cases}
\end{align}
\end{allowdisplaybreaks}
\end{subequations}

\subsection{Handling missing data}
One problem encountered when implementing the spike detection method on large sets of heart rate data provided by Crickles was that some recorded time series have portions of missing data, ranging between minutes to hours within a single recorded activity, resulting in unrealistic results, such as detecting an implausibly large amount of spikes.

To handle these issues, dummy heart rate readings were filled into the regions with missing data by linearly interpolating between the two points where the heart rate was last and next recorded, so that there are now heart rate readings every second in the time series. This was done before the heart rate time series was smoothed with a moving average. After this interpolated time series was smoothed, the regions where heart rate readings were added in for interpolation were removed so that the final smoothed version was aligned to the original heart rate data. The residuals were then obtained and a threshold was set on the residuals to detect spikes as explained before. This method was used instead of simply removing the empty regions from the data set as it ensured that sections of recorded heart rate separated by missing regions were unrelated, and not treated as one continuous activity. 

\FloatBarrier
\section{Example simulated time series} \label{SI: example simulations}
\setcounter{figure}{0}

We show here the potential surface used to simulate the base heart rate in Eqn.\ref{eqn: base heart rate SDE}. Following this we give two examples of simulated timeseries. 

\begin{figure}[ht]
    \centering
    \includegraphics[width=\linewidth]{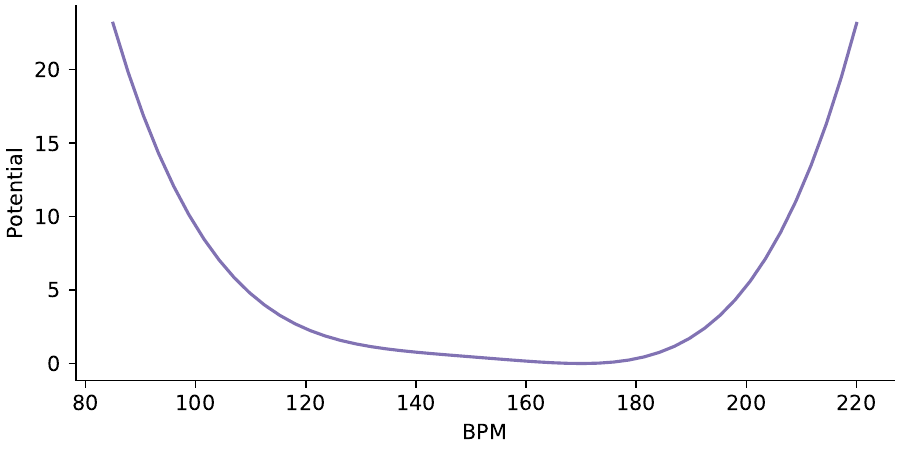}
    \caption{Example of a fixed potential surface as a function of heart rate. The asymmetry in the potential surface ensures that the heart rate does not frequently exceed a natural maximum, while allowing excursions to lower heart rates.}
    \label{fig:potential surface}
\end{figure}

\begin{figure}[ht]
    \centering
    \begin{subfigure}{\linewidth}
        \centering
        \includegraphics[width=\linewidth]{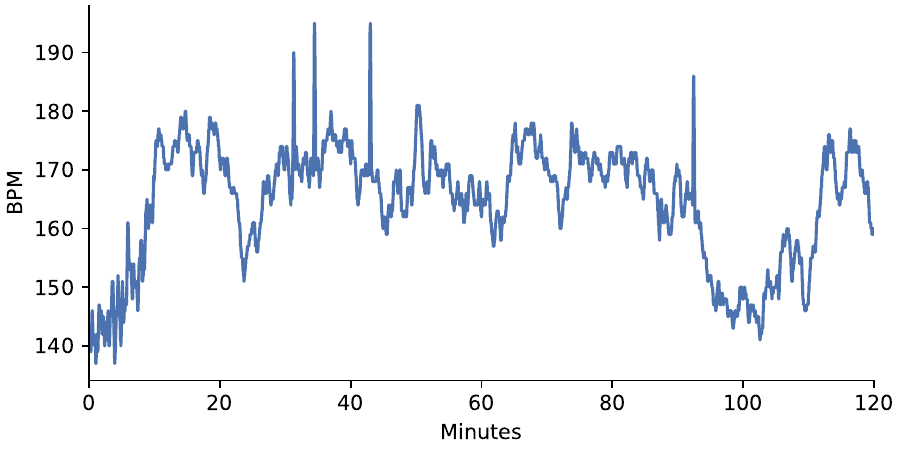}
        \caption{Random activity pattern}
        \label{fig:example simulated time series (constant)}
    \end{subfigure}
    \begin{subfigure}{\linewidth}
        \centering
        \includegraphics[width=\linewidth]{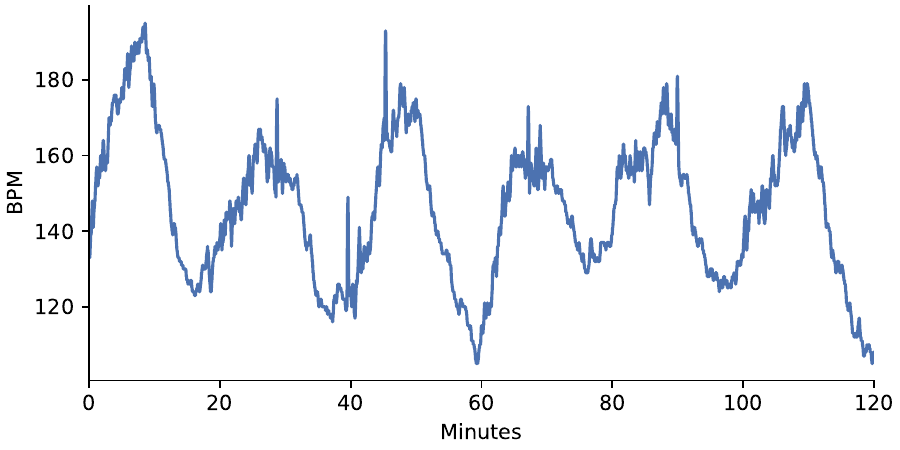}
        \caption{Interval training}
        \label{fig:example simulated time series (interval)}
    \end{subfigure}
    \caption{
        Example simulated time series with infrequent spikes. In (a) a constant potential surface is used, with additional noise during the first 10 minutes. In (b) the minimum of the potential surface is used to generate an interval training pattern. Additional noise is added in every other 10 minute section, beginning with the first 10 minutes.
    }
    \label{fig:example simulated time series}
\end{figure}

\FloatBarrier
\section{Poisson inference} \label{SI: Poisson}
\setcounter{figure}{0}

We describe here the likelihood ratio test used to consider the following hypotheses: 
\begin{itemize}
    \item $H_0$: spikes occur according to a Poisson process with constant intensity.
    
    \item $H_1$ spikes occur according to a Poisson process whose intensity is a piecewise constant function of the heart rate. 
\end{itemize}

We begin by calculating the MLE under $H_0$. Denote the total number of spikes across all an individual's time series by $N$, and the total length of the time series by $T$. The log-likelihood function under $H_0$ is 
\begin{equation}
    \ell_0(\lambda_0) = N\log(\lambda_0) - \lambda_0 \,T
\end{equation}
which is maximised by 
\begin{equation}
    \hat{\lambda_0} = \frac{N}{T}
\end{equation}
assuming of course that $T>0$.

We now consider the model under the alternative hypothesis. Denoted by $X(t)$ a smoothed heart rate time series with spikes removed. Divide a range of possible heart rates into $k$ intervals $[X_0,X_1),[X_1,X_2),\dots,[X_{k-1},X_k]$ with $X_0<X_1,\dots,X_k$. Define $N_i$ to be the number of spike events that are detected across the time series while $X_t\in[X_{i-1},X_i)$. Additionally define $T_i$ to be the amount of time for which $X_t\in[X_{i-1},X_i)$. Formally
\begin{align}
    N_i &= \sum_{j=1}^N \mathds{1}\{{X_t\in[X_{i-1},X_i)}\} \\
    T_i &= \int_0^T \mathds{1}\{{X_t\in[X_{i-1},X_i)}\} \, \text{dt}
\end{align}
The log-likelihood function is then
\begin{equation}
    \ell(\lambda_1,\dots,\lambda_k) = 
    \sum_{i=1}^k N_i \log(T_i) - T_i \lambda_i
\end{equation}
which can be maximised by maximising each component of the sum. This gives the MLE 
\begin{equation}
    (\hat{\lambda_1},\dots,\hat{\lambda_k}) = \bigg(\frac{N_1}{T_1},\dots,\frac{N_k}{T_k}\bigg)
\end{equation}
As the case of a homogeneous Poisson process is a restriction of this more complex model, the more complex model will always provide a `better' fit to the data. We wish to test if there is a significant relationship between the intensity of the Poisson process and the heart rate by determining if this more complex models provides a statistically significantly better fit.

\begin{remark}
In theory an arbitrary number of intervals could be included at physically unrealistic heart rates that would never be attained. To therefore avoid increasing $k$ unnecessarily and encountering issues in the definition of MLEs we assume that all $T_i>0$ and discount/merge any intervals for which this condition fails. 
\end{remark}

The test statistic is given by 
\begin{align*}
    D 
    &= 2\big[\ell(\hat{\lambda}) - \ell_0(\hat{\lambda_0})\big] \\
    &= 2\bigg[ \sum_{i=1}^k \Big(N_i\log\Big(\frac{N_i}{T_i}\Big) - 1\Big) - N\Big(\log\Big(\frac{N}{T}\Big) - 1\Big) \bigg]
\end{align*}
This is compared against a chi-squared distribution with $k-1$ degrees of freedom to test for significance. 

\FloatBarrier
\section{Spike density error} \label{SI: Spike density error}
\setcounter{figure}{0}

Some basic properties of the spike density error:
\begin{itemize}
    \item As it is the integral of a sum of non-negative functions it is always non-negative. 
    \item As there is no theoretical limit on the size or number of spikes the error is unbounded. 
    \item Errors increase if relatively larger spikes are missed or erroneously detected. 
\end{itemize}

We now present a proposition that demonstrates the relationship between accuracy and error behaves as expected. 

\begin{proposition}
Consider a situation in which there is only one true spike and one detected spike of equal size (assumed to be 1 for simplicity) with the detected spike a distance $h$ from the true spike. Let $s$ be the time of the true spike and $s+h$ the time of the detected spike. Assume the density kernel $f$ has finite support and that the location of the spikes are sufficiently far from the beginning or end of the time series that the supports of $f_s$ and $f_{s+h}$ lie entirely within $[0,T]$.
In this situation the spike density error decreases monotonically as $h$ decreases and is given by
\begin{equation} \label{eqn: sde prop}
    \epsilon(h) = \frac{2}{T} \int\limits_{-\frac{h}{2}}^{\frac{h}{2}} f(t) \,dt = \frac{4}{T} \int\limits_{0}^{\frac{h}{2}} f(t) \,dt
\end{equation}
\end{proposition}

\begin{proof}
Assume $h>0$. Recall the definition of the shifted and scaled density kernel (Eqn.\ref{eqn: scaled kernel density})
\begin{equation*}
    f_s (t) = \frac{f(t-s)}{\int_0^T f(\tau-s)\,d\tau}
\end{equation*}
Using assumptions about the spikes and density kernel this reduces to
\begin{equation*}
    f_s (t) = f(t-s)
\end{equation*}
Then define
\begin{align*}
    g(t) &= f_s(t) - f_{s+h}(t) \\
    &= f_s(t) - f_s(t-h)
\end{align*}
As the density kernel $f$ is increasing below $0$, $$f_s(t)\geq f_{s}(t-h)\,\forall\, t<s$$ and so $$g(t)\geq 0\,\forall\, t<s.$$
Similarly since $f$ is decreasing above $0$,
$$g(t)\leq 0\,\forall\, t>s+h.$$
In the region $s<t<s+h$, $f_s(t)$ is decreasing and $f_{s+h}(t)$ is increasing, hence there exists a point $t^{*}\in(s,s+h)$ such that $g(t^{*})=0$, $g(t)\geq0\,\forall\,t<t^{*}$ and $g(t)\leq0\,\forall\,t>t^{*}$. As $f$ is not required to be strictly increasing, such a point may not be unique. Since 
\begin{align*}
    f_s\bigg(s+\frac{h}{2}\bigg) 
    &= f_s\bigg(s-\frac{h}{2}\bigg) \\
    &= f_{s+h}\bigg(s-\frac{h}{2}-h\bigg) \\
    &= f_{s+h}\bigg(s+\frac{h}{2}\bigg)
\end{align*}
$t^{*}=s+\frac{h}{2}$ is a suitable point. 

Recall the definition of the spike density error (Eqn.\ref{eqn: spike density error})
\begin{equation*} 
    \epsilon = \frac{1}{T} \int_0^T \big|\rho_r(t) - \rho_d(t) \big|\,dt
\end{equation*}
We have that 
\begin{equation*}
    \big|\rho_r(t) - \rho_d(t) \big| = 
    \begin{cases}
    f_s(t) - f_{s+h}(t) \quad &\text{if $t<t^{*}$} \\
    0 \quad &\text{if $t=t^{*}$} \\
    f_{s+h}(t) - f_s(t) \quad &\text{if $t>t^{*}$}
    \end{cases}
\end{equation*}
hence
\begin{align*}
    \epsilon 
    =& \frac{1}{T} \int_0^T \big|\rho_r(t) - \rho_d(t) \big|\,dt \\
    =& \frac{1}{T}\int_0^{t^{*}} f_s(t) - f_{s+h}(t)\,dt + \frac{1}{T}\int_{t^{*}}^T f_{s+h}(t) - f_s(t)\,dt \\
    =& \frac{1}{T}\int_0^{t^{*}} f_s(t)\,dt - \frac{1}{T}\int_0^{t^{*}} f_s(t-h)\,dt \\
    &+ \frac{1}{T}\int_{t^{*}}^T f_s(t-h)\,dt - \frac{1}{T}\int_{t^{*}}^T f_s(t)\,dt \\
    =& \frac{1}{T}\int_0^{t^{*}} f_s(t)\,dt - \frac{1}{T}\int_0^{t^{*}-h} f_s(t)\,dt \\
    &+ \frac{1}{T}\int_{t^{*}-h}^T f_s(t)\,dt - \frac{1}{T}\int_{t^{*}}^T f_s(t)\,dt \\
    =& \frac{2}{T} \int_{t^{*}-h}^{t^{*}+h} f_s(t) \,dt
\end{align*}
Using $t^{*}=s+\frac{h}{2}$ then gives Eqn.\ref{eqn: sde prop}. As $f\geq0$ decreasing $h$ decreases the error. The same argument can be made for $h<0$ using the symmetry of $f$.

\end{proof}

\FloatBarrier
\section{Distribution of spike in heart rate} \label{SI: spikes at heart rate}
\setcounter{figure}{0}

\begin{figure}[H]
    \centering
    \includegraphics[width=\linewidth]{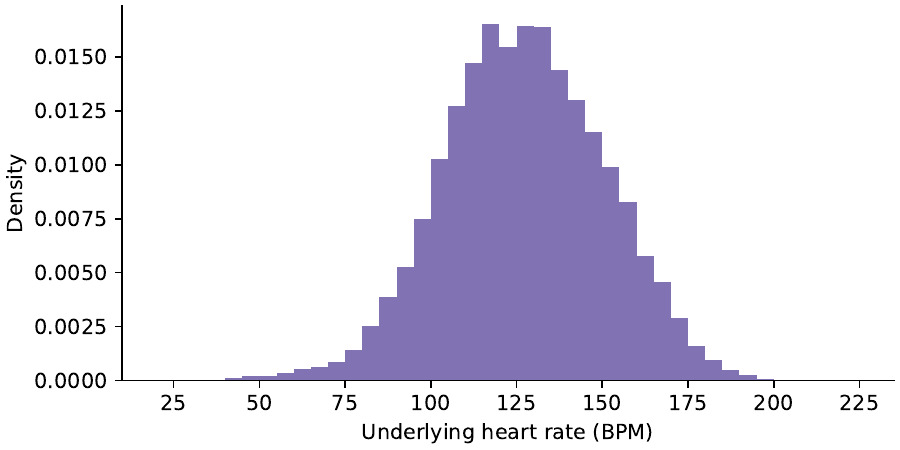}
    \caption{Histogram showing the distribution of heart rate where spikes were detected in the real heart rate data from Crickles. In order to account for different numbers of activities and spikes for different athletes, a separate histogram was generated for each athlete, and these histograms were then averaged to give the histogram here. The distribution shows a normal distribution-like curve, where the majority of the spikes were detected between ranges $105-140$ bpm.}
    \label{fig:distribution of spikes in hr}
\end{figure}

\FloatBarrier
\section{Distribution of spike times} \label{SI: spike times}
\setcounter{figure}{0}

\begin{figure}[H]
    \centering
    \includegraphics[width=\linewidth]{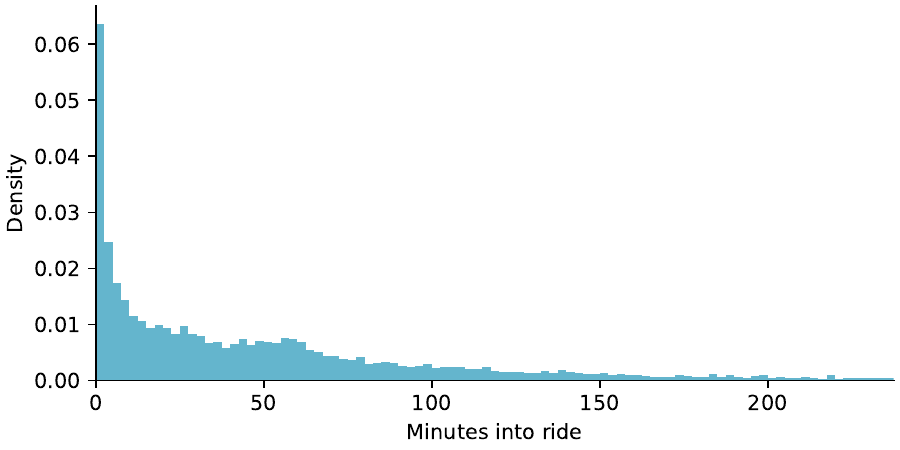}
    \caption{Histogram showing the distribution of times at which spikes were detected in the real heart rate data from Crickles. In order to account for different numbers of activities and spikes for different athletes, a separate histogram was generated for each athlete, and these histograms were then averaged to give the histogram here. This shows that the majority of spikes were detected within the first 10 minutes of the time series, which may be due to do the failure of sports devices to record an accurate heart rate during the start of an activity as discussed in section \ref{sec: further work}. }
    \label{fig:distribution of spike times}
\end{figure}

\end{document}